\documentclass[12pt]{article}
\usepackage{amsmath}
\usepackage{amssymb, amscd, amsthm,amsfonts}
\usepackage[all]{xy}
\usepackage[dvips]{graphicx}
\usepackage{verbatim}
\newtheorem{theo}[]{{\emph{Theorem}}}

\newtheorem{lemma}[]{{\emph{Lemma}}}

\theoremstyle{remark}
\newtheorem*{remark}{\textbf{Remark}}

\theoremstyle{definition}

\newcommand{\ra}{\rightarrow}

\newcommand{\calF}{\mathcal{F}} %added by hyq
\newcommand{\bF}{\mathbb{F}}
\newcommand{\cC}{\mathcal{C}}%short for \calC

\newcommand{\Tr}{\mathrm{Tr}}

\newcommand{\al}{\alpha}
\newcommand{\be}{\beta}
\newcommand{\ga}{\gamma}
\newcommand{\veps}{\varepsilon}
\newcommand{\ka}{\kappa}

\newcommand{\s}{\scriptstyle}

\newcommand{\om}{\omega}
\DeclareMathOperator{\Tra}{\mathrm Tr}

\begin{document}

\title{Cyclic Codes and Sequences from Kasami-Welch Functions} \maketitle

\begin{center}
$\mathrm{Jinquan\;\; Luo\qquad\quad San\;\; Ling \quad\qquad
and\qquad Chaoping\;\; Xing}$ \footnotetext{J.Luo is with the School
of Mathematics, Yangzhou University, Jiangsu Province, 225009, China
and with the Division of Mathematics, School of Physics and
Mathematical Sciences, Nanyang Technological University, Singapore.
\par S.Ling and C.Xing are with the Division of Mathematics, School
of Physics and Mathematical Sciences, Nanyang Technological
University, Singapore.
\par \quad E-mail addresses: jqluo@ntu.edu.sg, lingsan@ntu.edu.sg, xingcp@ntu.edu.sg.}
\end{center}
\newpage
 \textbf{Abstract} \par Let $q=2^n$, $0\leq k\leq n-1$ and $k\neq n/2$. In this paper we determine
the value distribution of following exponential sums
\[\sum\limits_{x\in
\bF_q}(-1)^{\Tra_1^n(\alpha x^{2^{3k}+1}+\beta
x^{2^k+1})}\quad(\alpha,\beta\in \bF_{q})\]
 and
\[\sum\limits_{x\in
\bF_q}(-1)^{\Tra_1^n(\alpha x^{2^{3k}+1}+\beta x^{2^k+1}+\ga
x)}\quad(\alpha,\beta,\ga\in \bF_{q})\]

 where $\Tra_1^n: \bF_{2^n}\ra \bF_2$ is the canonical trace
mapping. As applications:
 \begin{itemize}
    \item[(1).]  We determine the weight distribution
of the binary cyclic codes $\cC_1$ and $\cC_2$ with parity-check
polynomials $h_2(x)h_3(x)$ and $h_1(x)h_2(x)h_3(x)$ respectively
where  $h_1(x)$, $h_2(x)$ and $h_3(x)$ are the minimal polynomials
of $\pi^{-1}$, $\pi^{-(2^k+1)}$ and $\pi^{-(2^{3k}+1)}$ respectively
for a primitive element $\pi$ of $\bF_q$.
    \item[(2).]We determine the correlation distribution among
    a family of binary m-sequences. \end{itemize}

\newpage
\section{Introduction}

\quad Basic results on finite fields could be found in \cite{Lid Nie}. These notations are fixed throughout this paper except for specific statements.
\begin{itemize}
  \item Let $n$ be a positive integer, $q=2^n$,  $\bF_q$ be the finite field
  of order $q$. Let $\pi$ be a primitive element of $\bF_q$.
  \item Let $\Tra_i^j:\bF_{2^i}\ra\bF_{2^j}$ be the trace mapping,  and
$\chi(x)=(-1)^{\Tra_1^n(x)}$ be the canonical additive character on
$\bF_q$.
  \item Let $k$ be a positive integer, $1\leq k\leq n-1$ and  $k\notin \{\frac{n}{4},\frac{n}{2},\frac{3n}{4}\}$.
  Let $d=\gcd(n,k)$, $q_0=2^d$, and $s=n/d$.
  \item Let $m=n/2$ (if $n$ is even) and $\mu=(-1)^{m/d}$.
\end{itemize}

 For binary cyclic code $\cC$ with length $l$, let $A_i$ be the
number of codewords in $\cC$ with Hamming weight $i$. The weight
distribution $\{A_0,A_1,\cdots,A_l\}$ is an important research
object for both theoretical and application interests in coding
theory. Classical coding theory reveals that the weight of each
codeword can be expressed by binary exponential sums so that the
weight distribution of $\cC$ can be determined if the corresponding
exponential sums can be calculated explicitly (Kasami \cite{Kasa1},
\cite{Kasa2}, \cite{Kasa3},van der Vlugt \cite{Vand2},
Wolfmann\cite{Wolf}).

In general, let $q=2^n$, $\cC$ be the binary cyclic code with length
$l=q-1$ and parity-check polynomial
\[h(x)=h_1(x)\cdots h_u(x)\quad (u\geq 1)\]
where $h_i(x)$ $(1\leq i\leq u)$ are distinct irreducible
polynomials in $\bF_{2}[x]$ with the same degree $e_i$ $(1\leq i\leq
u)$, then $\mathrm{dim}_{\bF_{2}}\cC=\sum\limits_{i=1}^{u}e_i$ .Let
$\pi^{-s_i}$ be a zero of $h_i(x)$, $1\leq s_i\leq q-2$ $(1\leq
i\leq u).$ Then the codewords in $\cC$ can be expressed by
\[c(\alpha_1,\cdots,\alpha_u)=(c_0,c_1,\cdots,c_{l-1})\quad (\alpha_1,\cdots,\alpha_u\in \bF_q)\]
where
$c_i=\sum\limits_{\lambda=1}^{u}\Tra^n_{1}(\alpha_{\lambda}\pi^{is_{\lambda}})$
$(0\leq i\leq l-1)$. Therefore the Hamming weight of the codeword
$c=c(\alpha_1,\cdots,\alpha_u)$ is {\setlength\arraycolsep{2pt}
\begin{eqnarray} \label{Wei}
w_H\left(c\right)&=& \#\left\{i\left|0\leq i\leq l-1,c_i\neq
0\right.\right\}
\nonumber\\[1mm]
&=& l-\#\left\{i\left|0\leq i\leq l-1,c_i=0\right.\right\}
\nonumber\\[1mm]
&=&
l-\frac{1}{2}\,\sum\limits_{i=0}^{l-1}\left(1+(-1)^{\Tra_{1}^n\left(\sum\limits_{\lambda=1}^{u}\alpha_{\lambda}\pi^{is_{\lambda}}\right)}\right)
\nonumber\\[1mm]
&=&l-\frac{l}{2}-\frac{1}{2}\,\sum\limits_{x\in
\bF_q^*}(-1)^{\Tra_1^n(f(x))} \nonumber
\\[1mm]
&=&\frac{l}{2}+\frac{1}{2}-\frac{1}{2}\,S(\alpha_1,\cdots,\alpha_u)\nonumber\\[1mm]
&=&2^{n-1}-\frac{1}{2}\,S(\alpha_1,\cdots,\alpha_u)
\end{eqnarray}
} where
$f(x)=\alpha_1x^{s_1}+\alpha_2x^{s_2}+\cdots+\alpha_ux^{s_u}\in
\bF_{q}[x]$, $\bF_q^*=\bF_q\backslash\{0\}$,  and
\[S(\alpha_1,\cdots,\alpha_u)=\sum\limits_{x\in \bF_q}(-1)^{\Tra_1^n(\alpha_1x^{s_1}+\cdots+\alpha_ux^{s_u})}.\]
In this way, the weight distribution of cyclic code $\cC$ can be
derived from the explicit evaluating of the exponential sums
\[S(\alpha_1,\cdots,\alpha_u)\quad(\alpha_1,\cdots,\alpha_u\in \bF_q).\]

 Let $h_1(x)$, $h_2(x)$ and $h_{3}(x)$
be the minimal polynomials of $\pi^{-1},\pi^{-(2^k+1)}$ and
$\pi^{-{(2^{3k}+1)}}$ over $\bF_{2}$ respectively. Then
$\mathrm{deg}\,h_i(x)=n\; \text{for}\; i=1,2,3.$

 Let
$\cC_1$ and $\cC_2$ be the binary cyclic codes  with length $l=q-1$
and parity-check polynomials $h_2(x)h_3(x)$ and $h_1(x)h_2(x)h_3(x)$
respectively. It is a consequence that $\cC_1$($\cC_2$, resp.) is
the dual of the binary BCH code with designed distance $5$ ($7$,
resp.) whose zeroes include $\pi^{2^{3k}+1}$ and $\pi^{2^k+1}$
($\pi^{2^{3k}+1}$ ,$\pi^{2^k+1}$ and $\pi$, resp.). Then we know
that the dimensions of $\cC_1$ and $\cC_2$  are $2n$ and $3n$
respectively.

For $\al,\be,\ga\in \bF_q$,  define the exponential sums
\begin{equation}\label{def T}
T(\al,\be)=\sum\limits_{x\in \bF_q}(-1)^{\Tra_1^n(\alpha
x^{2^{3k}+1}+\beta x^{2^k+1})}
\end{equation}
and
\begin{equation}\label{def S}
S(\al,\be,\ga)=\sum\limits_{x\in \bF_q}(-1)^{\Tra_1^n(\alpha
x^{2^{3k}+1}+\beta x^{2^k+1}+\ga x)}.
\end{equation}
 Then the complete
weight distributions of $\cC_1$ and $\cC_2$ can be derived from the
explicit evaluation of $T(\al,\be)$ and $S(\al,\be,\ga)$.

 Another application of binary exponential sums is to obtain the cross correlation and auto-correlation distribution
 among binary sequences.
 Let $\mathcal{F}$ be a collection of binary
m-sequences of period $q-1$ defined by

\[\mathcal{F}=\left\{\left\{a_i(t)\right\}_{i=0}^{q-2}|\,0\leq i\leq L-1 \right\}\]

The \emph{correlation function} of $a_i$ and $a_j$ for a shift
$\tau$ is defined by

\[M_{{i},{j}}(\tau)=\sum\limits_{\lambda=0}^{q-2}(-1)^{a_i({\lambda})-a_j({\lambda+\tau})}\hspace{2cm}(0\leq \tau\leq
q-2).\]

Binary sequences with low cross correlation and auto-correlation are
widely used in Code Division Multiple Access(CDMA) spread
spectrum(see Paterson\cite{Part},  Simon, Omura and Scholtz\cite{Sim
Omu}). Pairs of binary m-sequences with few-valued auto and cross
correlations have been extensively studied for several decades, see
Canteaut, Charpin and Dobbertin \cite{Can Cha}, Cusick and Dobbertin
\cite{Cus}, Ding, Helleseth and Lam\cite{Din Hel1}, Ding, Helleseth
and Martinsen\cite{Din Hel2}, Dobbertin, Felke, Helleseth and
Rosendahl \cite{Dob Fel},  Gold \cite{Gold}, Helleseth
\cite{Hell1},\cite{Hell2}, Helleseth, Kholosha and Ness\cite{Hel
Kho}, Helleseth and Kumar \cite{Hel Kum}, Hollmann and Xiang
\cite{Hol Xia}, Ness and Helleseth\cite{Nes Hel}-\cite{Nes Hel2},
Niho\cite{Niho}, Rosendahl\cite{Rose}, Yu and Gong\cite{Yu
Gon1}-\cite{Yu Gon2} and references therein.

Define the collection of sequences

\[
  \calF_1=\left\{a_{\al,\be}=\left(a_{\al,\be}(\pi^{\lambda})\right)_{\lambda=0}^{q-2}\big{|}\,\al, \be\in \bF_{q} \right\}
\]
where
  $a_{\al,\be}(\pi^{\lambda})=\Tra_1^n(\al \pi^{\lambda(2^{3k}+1)}+\be
  \pi^{\lambda(2^k+1)}+\pi^{\lambda})$.

 If $n/d$  is odd, define
\[
  \calF_2=\left\{a_{\be}=\left(a_{\al}(\pi^{\lambda})\right)_{\lambda=0}^{q-2}\big{|}\,\be\in \bF_q \right\}\bigcup \left\{a=\left(a(\pi^{\lambda})\right)_{\lambda=0}^{q-2}\right\}
\]
where
  $a_{\al}(\pi^{\lambda})=\Tra_1^n(\al\pi^{\lambda(2^{3k}+1)}+
  \pi^{\lambda(2^k+1)})$ and $a(\pi^{\lambda})=\Tra_1^n(\pi^{\lambda(2^{3k}+1)})$.

In this correspondence we will study the following collection of
m-sequences with period $q-1$
\begin{equation}\label{def F}
\calF=\left\{
\begin{array}{ll}
\calF_1, &\text{if}\; n/d\;\text{is even}\\[1mm]
\calF_1\cup \calF_2, &\text{if}\; n/d\;\text{is odd.}
\end{array}
\right.
\end{equation}

This paper is presented as follows. In Section 2 we introduce some
preliminaries and give auxiliary results. In Section 3 we will give
the value distribution of $T(\al,\be)$ for $\al\in \bF_{2^m},\be\in
\bF_q$ and the weight distribution of $\cC_1$. In Section 3 we will
determine the value distribution of $S(\al,\be,\ga)$,  the
correlation distribution among the sequences in $\calF$, and then
the weight distribution of $\cC_2$. Most proofs of lemmas and
theorems are presented in several appendices. The main tools are
quadratic form theory over finite fields of characteristic 2, some
moment identities on $T(\alpha,\beta)$ and a class of Artin-Schreier
curves. This paper is binary analogue of \cite{Luo Lin Xin} and
\cite{Zen Hu Jia Yue Cao}.

\section{Preliminaries}

\quad We follow the notations in Section 1. The first machinery to
determine the values of exponential sums $T(\alpha,\beta)$
$(\alpha\in \bF_{p^m},\beta\in \bF_q)$ defined in (\ref{def T}) is
quadratic form theory over $\bF_{q_0}$.

 Let $H$
be an $s\times s$ matrix over $\bF_{q_0}$. For the quadratic form
\begin{equation}\label{qua for}
F:\bF_{q_0}^s\ra \bF_{q_0},\quad F(x)=XHX^T\quad
(X=(x_1,\cdots,x_s)\in \bF_{q_0}^s),
\end{equation}
define $r_F$ of $F$ to be the rank of the skew-symmetric matrix
$H+H^T$. Then $r_F$ is even.

 We have the following result on the exponential sum of binary quadratic forms.
\begin{lemma}\label{qua}

For the quadratic form $F=XHX^T$ defined in (\ref{qua for}),
\[\sum\limits_{X\in\bF_{q_0}^s}\zeta_p^{\Tra_1^{d}(F(X))}=\pm
q_0^{s-\frac{r_F}{2}}\; \text{or}\; 0\] Moreover, if $r_F=s$, then
\[\sum\limits_{X\in\bF_{q_0}^s}\zeta_p^{\Tra_1^{d}(F(X))}=\pm
q_0^{\frac{s}{2}}\]
\end{lemma}

\begin{proof}
We can calculate
\[{\setlength\arraycolsep{2pt}
\begin{array}{ll}&\left(\sum\limits_{X\in\bF_{q_0}^s}(-1)^{\Tra_{1}^{d}(F(X))}\right)^2
              =\sum\limits_{X,Z\in
\bF_{q_0}^s}(-1)^{\Tra_{1}^{d}\left(XHX^T+(X+Z)H(X+Z)^T\right)}\\[4mm]
              &\qquad\qquad=\sum\limits_{Z\in
\bF_{q_0}^s}(-1)^{\Tra_{1}^{d}(ZHZ^T)}\sum\limits_{X\in
\bF_{q_0}^s}(-1)^{\Tra_{1}^{d}\left(Z(H+H^T)X^T\right)}
\end{array}
}
\]
The inner sum is zero unless $Z(H+H^T)=0$. Define
\[\mathcal{Z}=\left\{Z\in \bF_{q_0}^s\,|\, Z(H+H^T)=0\right\}.\]
Then the map
\begin{equation}\label{map}
\begin{array}{rcl}
\eta:\mathcal{Z}&\longrightarrow& \bF_2\\[2mm]
Z&\mapsto& \Tra_1^d(ZHZ^T)
\end{array}
\end{equation}
is an additive group homomorphism.

If $\eta$ is surjective, then there are exactly one half $z\in
\mathcal{Z}$ mapping to $0$ and $1$ respectively. Hence
$\sum\limits_{X\in\bF_{q_0}^s}(-1)^{\Tra_{1}^{d}(F(X))}=0$.
Otherwise $\mathrm{Im}(\eta)=\left\{0\right\}$ and
$\left(\sum\limits_{X\in\bF_{q_0}^s}(-1)^{\Tra_{1}^{d}(F(X))}\right)^2=
q_0^{s}\cdot q_0^{s-r_F}$. Hence
$\sum\limits_{X\in\bF_{q_0}^s}(-1)^{\Tra_{1}^{d}(F(X))}=\pm
q_0^{s-\frac{r_F}{2}}$.

If $r_F=s$, then $\mathcal{Z}=\{0\}$ and
$\mathrm{Im}(\eta)=\left\{0\right\}$. Therefore
$\sum\limits_{X\in\bF_{q_0}^s}(-1)^{\Tra_{1}^{d}(F(X))}=\pm
q_0^{\frac{s}{2}}$.

\end{proof}

The following result, which has been proven in \cite{Lid Nie}, Chap.
6, will be used in Section 4.
\begin{lemma}\label{det gamma}
For the fixed quadratic form defined in  (\ref{qua for}), the value
distribution of
$\sum\limits_{X\in\bF_{q_0}^s}(-1)^{\Tra_1^{d}(F(X)+AX^T)}$ when $A$
runs through $\bF_{q_0}^s$ is shown as following
\[
\begin{array}{ccc}
value & \qquad\qquad\qquad&multiplicity \\[2mm]
0&\qquad\qquad\qquad&q_0^s-q_0^{r_F}\\[2mm]
q_0^{s-\frac{r_F}{2}}&\qquad\qquad\qquad &\frac{1}{2}(q_0^{r_F}+q_0^{\frac{r_F}{2}})\\[2mm]
-q_0^{s-\frac{r_F}{2}}&\qquad\qquad\qquad &\frac{1}{2}(q_0^{r_F}-q_0^{\frac{r_F}{2}})\\[2mm]
\end{array}
\]
\end{lemma}

Since $s=n/d$, the field $\bF_q$ is a vector space over $\bF_{q_0}$
with dimension $s$. We fix a basis $v_1,\cdots,v_s$ of $\bF_q$ over
$\bF_{q_0}$. Then each $x\in \bF_q$ can be uniquely expressed as
\[x=x_1v_1+\cdots+x_sv_s\quad (x_i\in \bF_{q_0}).\]
Thus we have the following $\bF_{q_0}$-linear isomorphism:
\[\bF_q\xrightarrow{\sim}\bF_{q_0}^s,\quad x=x_1v_1+\cdots+x_sv_s\mapsto
X=(x_1,\cdots,x_s).\] With this isomorphism, a function $f:\bF_q\ra
\bF_{q_0}$ induces a function $F:\bF_{q_0}^s\ra \bF_{q_0}$ where for
$X=(x_1,\cdots,x_s)\in \bF_{q_0}^s, F(X)=f(x)$ with
$x=x_1v_1+\cdots+x_sv_s$. In this way, function
$f(x)=\Tra_{d}^n(\gamma x)$ for $\gamma\in \bF_q$ induces a linear
form \begin{equation} F(X)=\Tra_{d}^n(\gamma
x)=\sum\limits_{i=1}^{s}\Tra_{d}^n(\gamma v_i)x_i=A_{\ga}X^T
\end{equation}\label{def A_gamma}
 where $A_{\ga}=\left(\Tra_{d}^n(\gamma
v_1),\cdots,\Tra_{d}^n(\gamma v_s)\right),$
 and
$f_{\alpha,\beta}(x)=\Tra_d^n(\alpha x^{p^{3k}+1}+\beta x^{p^k+1})$
for $\al,\be\in \bF_q$ induces a quadratic form

\begin{eqnarray}\label{def H_al be}
F_{\alpha,\beta}(X)=XH_{\alpha,\beta}X^T
\end{eqnarray}

From Lemma \ref{qua},  for $\alpha,\beta,\ga\in \bF_q$, in order to
determine the values of
\[T(\alpha,\beta)=\sum\limits_{x\in \bF_q}(-1)^{\Tra_1^n(\alpha
x^{2^{3k}+1}+\beta x^{2^k+1})}=\sum\limits_{X\in
\bF_{q_0}^s}(-1)^{\Tra_1^{d}\left(XH_{\alpha,\beta}X^T\right)}\] and
\[S(\alpha,\beta,\ga)=\sum\limits_{x\in \bF_q}(-1)^{\Tra_1^n(\alpha
x^{2^{3k}+1}+\beta x^{2^k+1}+\ga x)}=\sum\limits_{X\in
\bF_{q_0}^s}(-1)^{\Tra_1^{d}\left(XH_{\alpha,\beta}X^T+A_{\ga}X^T\right)},\]
we need to determine the rank of $H_{\alpha,\beta}+H_{\al,\be}^T$
over $\bF_{q_0}$.

Define $d'=\gcd(n,2k)$. Then an easy observation shows
\begin{equation}\label{rel d d'}
d'=\left\{
\begin{array}{ll}
d, & \text{if}\; n/d\;\text{is odd};\\[1mm]
2d, &\text{otherwise.}
\end{array}
\right.
\end{equation}

\begin{lemma}\label{rank}
For $(\alpha,\beta)\in \bF_{p^m}\times \bF_q\backslash\{(0,0)\}$,
let $r_{\alpha,\beta}$ be the rank of
$H_{\alpha,\beta}+H_{\alpha,\beta}^T$.  Then we have
\begin{itemize}
  \item[(i).] if $n/d$ is odd, then  then the possible values of $r_{\al,\be}$ are $s-1$ and $s-3$.
  \item[(ii).] if $n/d$ is even, then the possible values of $r_{\al,\be}$
  are $s$, $s-2$, $s-4$ and $s-6$.
\end{itemize}
\end{lemma}

 In order to determine the value distribution of
$T(\al,\be)$
 for $\al,\be\in \bF_q$, we need the
following result on moments of $T(\al,\be)$ and $S(\al,\be,\ga)$.
\begin{lemma}\label{moment}
For the exponential sum $T(\al,\be)$ and $S(\al,\be,\ga)$,
\[\begin{array}{ll}&(i). \;\;\sum\limits_{\al,\be\in
\bF_q}T(\al,\be)=2^{2n};\\[2mm]
                   &(ii).\; \text{if}\;n/d\;\text{is even},
                   \;\text{then}\\[2mm]
&\quad\quad\quad\sum\limits_{\al,\be\in
\bF_q}T(\al,\be)^2= (2^{n+d}+2^n-2^d)\cdot 2^{2n}\\[2mm]
                   &(iii). \;\text{if}\;n/d\;\text{is even},\; \text{then}\\[2mm]
                   &\quad\quad\quad \sum\limits_{\al,\be\in
\bF_q}T(\al,\be)^3=(2^{n+3d}+2^n-2^{3d})\cdot 2^{2n}.\\[4mm]
&(iv). \;\sum\limits_{\al,\be,\ga\in
\bF_q}S(\al,\be,\ga)^3=(2^{n+d}+2^n-2^{d})\cdot 2^{3n}.
\end{array}
\]
\end{lemma}
\begin{proof}
see  \textbf{Appendix A.}
\end{proof}
\begin{remark}
We could calculate the moment identities $\sum\limits_{\al,\be\in
\bF_q}T(\al,\be)^i$ and $\sum\limits_{\al,\be,\ga\in
\bF_q}T(\al,\be,\ga)^i$ for $1\leq i\leq 3$. But the result shown above is enough to determine the value distributions
of $T(\al,\be)$ and $S(\al,\be,\ga)$.
\end{remark}
In the case $n/d$ is even, we could determine the explicit values of
$T(\al,\be)$. To this end we will study a class of Artin-Schreier
curves. A similar technique has been applied in Coulter
\cite{Coul1}, Theorem 6.1.

\begin{lemma}\label{Artin}Suppose $(\al,\be)\in (\bF_{q}\times
\bF_q)\big{\backslash}\{0,0\}$ and $d'=2d$. Let $N$ be the number of
$\bF_q$-rational (affine) points on the curve
\begin{equation}\label{Artin Sch}
\al x^{2^{3k}+1}+\be x^{2^k+1}=y^{2^d}+y.
\end{equation}
 Then
\[N=q+(2^d-1)\cdot T(\al,\be).\]
\end{lemma}
\begin{proof}
see \textbf{Appendix A.}
\end{proof}

Now we give an explicit evaluation of $T(\al,\be)$ in the case $n/d$
is even.
\begin{lemma}\label{reduce num}
Assumptions as in Lemma \ref{Artin} and let $n=2m$, then
\[
T(\al,\be)=\left\{
\begin{array}{ll}
\mu 2^{m}, &\text{if}\; r_{\al,\be}=s\\[2mm]
-\mu 2^{m+d}, &\text{if}\; r_{\al,\be}=s-2\\[2mm]
\mu 2^{m+2d}, &\text{if}\; r_{\al,\be}=s-4\\[2mm]
-\mu 2^{m+3d}, &\text{if}\; r_{\al,\be}=s-6.
\end{array}
\right.
\]
where $\mu=(-1)^{m/d}$.
\end{lemma}
\begin{proof}
Consider the $\bF_q$-rational (affine) points on the Artin-Schreier
curve in Lemma \ref{Artin}. It is easy to verify that $(0,y)$ with
$y\in \bF_{2^d}$ are exactly the points on the curve with $x=0$. If
$(x,y)$ with $x\neq 0$ is a point on this curve, then so are $(t
x,t^{2^d+1}y)$ with $t^{2^{2d}-1}=1$ (note that $2^{3k}+1\equiv
2^k+1\equiv 2^d+1\pmod {2^{2d}-1}$ since $3k/d$ and $k/d$ are both
odd by (\ref{rel d d'})). In total, we have
\[q+(2^d-1)T(\al,\be)=N\equiv 2^d\pmod {2^{2d}-1}\]
which yields
\[T(\al,\be)\equiv 1\pmod{2^d+1}.\]

 We only consider the case $r_{\al,\be}=s$ and $m/d$ is odd. The other cases are similar. In this case $T(\al,\be)=\pm
 2^{m}$. Assume $T(\al,\be)=2^{m}$.
 Then ${2^d+1}\mid 2^{m}-1$ which contradicts to $m/d$ is odd. Therefore $T(\al,\be)=-2^{m}$.
\end{proof}
\begin{remark}
(i) In the case $n/d$ is even, the non-binary case (see \cite{Luo
Lin Xin}, Lemma 6)  and the binary case could be solved in a unified
way. But the case $n/d$ odd is quite different.

 (ii). Applying Lemma
\ref{reduce num} to Lemma \ref{Artin}, we could determine the number
of rational points on the curve (\ref{Artin Sch}).
\end{remark}

\section{Exponential Sums $T(\al,\be)$ and Cyclic Code $\cC_1$}

\quad Define
\[N_{i}=\left\{(\al,\be)\in
\bF_{q}\times
\bF_q\backslash\{(0,0)\}\left|r_{\al,\be}=s-i\right.\right\}.
\] Then $n_i=\big{|}N_i\big{|}$.

According to the possible values of $T(\al,\be)$ given by Lemma
\ref{qua} (setting $F(X)=XH_{\al,\be}X^T=\Tra_{d}^n(\al
x^{2^{3k}+1}+\be x^{2^k+1})$), we define that for $\varepsilon=\pm
1$,
\[N_{\varepsilon,i}=\left\{(\al,\be)\in
\bF_q^2\backslash\{(0,0)\}\left|T(\al,\be)=\veps 2^{\frac{n+id}{2}}
\right.\right\}\]  and $n_{\veps,i}=|N_{\veps,i}|$.

 Now we could give the value distribution of $T(\al,\be)$ and the
 weight distribution of $\cC_1$. We will only focus on the case
 $n/d$ is even since the odd case has been solved in \cite{Hol Xia}.

\begin{theo}\label{value dis T}
The value distribution of the multi-set
$\left\{T(\al,\be)\left|\al,\be\in \bF_q\right.\right\}$ and the
weight distribution of $\cC_1$ are shown as following (Column 1 is
the value of $T(\al,\be)$, Column 2 is the weight of
$c(\al,\be)=\Tra_1^n\left(\al \pi^{i(2^{3k}+1)}+\be
\pi^{i(2^k+1)}\right)_{i=0}^{q-2}$ and Column 3 is the corresponding
multiplicity).

(i). For the case $n/d$ is odd,
\begin{center}
\begin{tabular}{|c|c|c|}
\hline
values &weights &multiplicity \\[2mm]
\hline
$2^{(n+d)/2}$&$2^{n-1}-2^{(n+d-2)/2}$&$\left(2^{n-d-1}+2^{(n-d-2)/2}\right)(2^n-1)$
\\[2mm]
\hline
$-2^{(n+d)/2}$&$2^{n-1}+2^{(n+d-2)/2}$&$\left(2^{n-d-1}-2^{(n-d-2)/2}\right)(2^n-1)$
\\[2mm]
\hline $0$&$2^{n-1}$&$\left(2^n-2^{n-d}+1\right)(2^n-1)$
\\[2mm]
\hline $2^n$&$0$&$1$
\\[2mm]
\hline
\end{tabular}
\end{center}

(ii).  For the case $n/d$ is even,
\begin{center}
\begin{tabular}{|c|c|c|}
\hline
values &weights& multiplicity \\[2mm]
\hline $\mu 2^m$& $2^{n-1}-\mu 2^{m-1}$& $
\frac{(2^n-1)\left(2^{n+6d}-2^{n+4d}-2^{n+d}+\mu 2^{m+5d}-\mu
2^{m+4d}+2^{6d}\right)}{(2^d+1)(2^{2d}-1)(2^{3d}+1)} $
\\[2mm]
\hline $-\mu {2}^{m+d}$& $2^{n-1}+\mu 2^{m+d-1}$& $
\frac{(2^n-1)(2^{n+3d}+2^{n+2d}-2^n-2^{n-d}-2^{n-2d}-\mu
2^{m+3d}+\mu 2^{m}+2^{3d})}{(2^d+1)^2(2^{2d}-1)} $
\\[2mm]
\hline
 $\mu {2}^{m+2d}$&$2^{n-1}-\mu 2^{m+2d-1}$&$
 \frac{
(2^{m-d}+\mu)(2^{m+d}+2^m-2^{m-2d}-\mu
2^d)(2^{n}-1)}{(2^d+1)^3(2^{d}-1)} $
\\[2mm]
\hline $-\mu {2}^{m+3d}$&$2^{n-1}+\mu 2^{m+3d-1}$&$
\frac{(2^{m-2d}-\mu)(2^{m-d}+\mu)(2^{n}-1)}{(2^d+1)(2^{2d}-1)(2^{3d}+1)}
$
\\[2mm]

\hline $2^n$&$0$&$1$
\\[2mm]
\hline
\end{tabular}
\end{center}
where $\mu=(-1)^{m/d}$.
\end{theo}

\begin{proof}
see \cite{Hol Xia} for (i) and \textbf{Appendix B} for (ii).
\end{proof}
\begin{remark}

\begin{itemize}
\item[(i).] In the case $k\in\left\{\frac{n}{4},\frac{n}{2},\frac{3n}{4}\right\}$, the exponential sum $T(\al,\be)=\sum\limits_{x\in
\bF_q}(-1)^{\Tr_1^n\left(\left(\al^{2^k}+\be\right)x^{2^{k}+1}\right)}$
and the cyclic code $\cC_1$ has been extensively studied, for
example, see \cite{Coul1}, \cite{Gold}.
\item[(ii).] In the case $k\in\left\{\frac{n}{6},\frac{5n}{6}\right\}$, the exponential sum $T(\al,\be)
$ and the weight distribution of $\cC_1$ has been shown in \cite{Luo
Tan Wan}. In the case $k\in \{\frac{n}{3},\frac{2n}{3}\}$, the
exponential sum
 $T(\al,\be)=\sum\limits_{x\in \bF_q}(-1)^{\Tr_1^n\left(\al^{2^{n-1}} x+\be x^{2^{k}+1}\right)}$
is a special case in \cite{Coul1}, \cite{Gold}.
    \item[(iii).] In the case $n/d$ is even. Since $\gcd(2^{n}-1,
2^k+1)=2^d+1$, the first $l'=\frac{2^n-1}{2^d+1}$ coordinates of
each codeword of $\cC_1$ form a cyclic code $\cC_1'$ with length
$l'$ and dimension $2n$. Let $(A_0',\cdots,A_{l'}')$ be the weight
distribution of $\cC_1'$, then $A_i'=A_{(2^d+1)i}$ $(0\leq i\leq
l')$.
\end{itemize}
\end{remark}

\section{Results on Correlation Distribution of Sequences and Cyclic Code $\cC_2$}

\quad Recall $\phi_{\al,\be}(x)$ in the proof of Lemma \ref{rank}
and $N_{i,\veps}$ in the proof of Theorem \ref{value dis T}. Finally
we will determine the value distribution of $S(\al,\be,\ga)$, the
correlation distribution among sequences in $\calF$ defined in
(\ref{def F}) and the weight distribution of $\cC_2$ defined in
Section 1.

\begin{theo}\label{value dis S}
The value distribution of the multi-set
$\left\{S(\al,\be,\ga)\left|\al,\be,\ga\in \bF_q\right.\right\}$ and
the weight distribution of $\cC_2$ are shown as following (Column 1
is the value of $S(\al,\be,\ga)$, Column 2 is the weight of
$c(\al,\be,\ga)=\left(\Tra_1^n(\al \pi^{i(2^{3k}+1)}+\be
\pi^{i(2^k+1)}+\ga \pi^i)\right)_{i=0}^{q-2}$ and Column 3 is the
corresponding multiplicity).

 (i). For the case $n/d$ is odd (that is, $d'=d$),
\begin{center}
\begin{tabular}{|c|c|c|}
\hline
values & weights & multiplicity \\[2mm]
\hline $2^{(n+d)/2}$ & $2^{n-1}-2^{(n+d-2)/2}$ &
$\frac{(2^{n-d-1}+2^{(n-d-2)/2})(2^n-1)(2^{n+2d}-2^{n}-2^{n-d}+2^{2d})}{2^{2d}-1}$
\\[2mm]
\hline $
-2^{(n+d)/2}$&$2^{n-1}+2^{(n+d-2)/2}$&$\frac{(2^{n-d-1}-2^{(n-d-2)/2})(2^n-1)(2^{n+2d}-2^{n}-2^{n-d}+2^{2d})}{2^{2d}-1}$
\\[2mm]
\hline
$2^{(n+3d)/2}$&$2^{n-1}-2^{(n+3d-2)/2}$&$\frac{(2^{n-3d-1}+2^{(n-3d-2)/2})(2^{n-d}-1)(2^n-1)}{2^{2d}-1}$
\\[2mm]
\hline
$-2^{(n+3d)/2}$&$2^{n-1}+2^{(n+3d-2)/2}$&$\frac{(2^{n-3d-1}-2^{(n-3d-2)/2})(2^{n-d}-1)(2^n-1)}{2^{2d}-1}$
\\[2mm]
\hline
$0$&$2^{n-1}$&$\scriptstyle{(2^n-1)(2^{2n}-2^{2n-d}+2^{2n-4d}+2^n-2^{n-d}-2^{n-3d}+1)}$
\\[2mm]
\hline $2^n$&$0$&$1$
\\[2mm]
\hline
\end{tabular}
\end{center}

(ii).  For the case $n/d$ is even,
\begin{center}
\begin{tabular}{|c|c|c|}
\hline
values &weights& multiplicity \\[2mm]
\hline $ \s{2^m}$& $\s{2^{n-1}- 2^{m-1}}$& $
\s{\frac{(2^{n-1}+2^{m-1})(2^n-1)\left(2^{n+6d}-2^{n+4d}-2^{n+d}+\mu
2^{m+5d}-\mu 2^{m+4d}+2^{6d}\right)}{(2^d+1)(2^{2d}-1)(2^{3d}+1)}} $
\\[2mm]
\hline $\s{- 2^m}$& $\s{2^{n-1}+ 2^{m-1}}$& $
\s{\frac{(2^{n-1}-2^{m-1})(2^n-1)\left(2^{n+6d}-2^{n+4d}-2^{n+d}+\mu
2^{m+5d}-\mu 2^{m+4d}+2^{6d}\right)}{(2^d+1)(2^{2d}-1)(2^{3d}+1)}} $
\\[2mm]
\hline

 $ \s{{2}^{m+d}}$& $\s{2^{n-1}- 2^{m+d-1}}$& $
\s{\frac{(2^{n-2d-1}+2^{m-d-1})(2^n-1)(2^{n+3d}+2^{n+2d}-2^n-2^{n-d}-2^{n-2d}-\mu
2^{m+3d}+\mu 2^{m}+2^{3d})}{(2^d+1)^2(2^{2d}-1)}} $
\\[2mm]
\hline

 $\s{- {2}^{m+d}}$& $\s{2^{n-1}+ 2^{m+d-1}}$& $
\s{\frac{(2^{n-2d-1}-2^{m-d-1})(2^n-1)(2^{n+3d}+2^{n+2d}-2^n-2^{n-d}-2^{n-2d}-\mu
2^{m+3d}+\mu 2^{m}+2^{3d})}{(2^d+1)^2(2^{2d}-1)}} $
\\[2mm]
\hline

 $ \s{{2}^{m+2d}}$&$\s{2^{n-1}- 2^{m+2d-1}}$&$
 \s{\frac{
(2^{m-d}+\mu)(2^{m+d}+2^m-2^{m-2d}-\mu
2^d)(2^{n-4d-1}+2^{m-2d-1})(2^{n}-1)}{(2^d+1)^3(2^{d}-1)}} $
\\[2mm]
\hline $\s{- {2}^{m+2d}}$&$\s{2^{n-1}+ 2^{m+2d-1}}$&$
 \s{\frac{
(2^{m-d}+\mu)(2^{m+d}+2^m-2^{m-2d}-\mu
2^d)(2^{n-4d-1}-2^{m-2d-1})(2^{n}-1)}{(2^d+1)^3(2^{d}-1)}} $
\\[2mm]
\hline
 $ \s{{2}^{m+3d}}$&$\s{2^{n-1}- 2^{m+3d-1}}$&$
\s{\frac{(2^{m-2d}-\mu)(2^{m-d}+\mu)(2^{n-6d-1}+2^{m-3d-1})(2^{n}-1)}{(2^d+1)(2^{2d}-1)(2^{3d}+1)}}
$
\\[2mm]

\hline
 $\s{- {2}^{m+3d}}$&$\s{2^{n-1}+ 2^{m+3d-1}}$&$
\s{\frac{(2^{m-2d}-\mu)(2^{m-d}+\mu)(2^{n-6d-1}-2^{m-3d-1})(2^{n}-1)}{(2^d+1)(2^{2d}-1)(2^{3d}+1)}}
$
\\[2mm]
\hline $\s{0}$& $ \s{2^{n-1}}$&$
\s{\frac{(2^n-1)(2^{2n}+2^{2n-9d}-\veps(
2^{3m}-
2^{3m-d}- 2^{3m-3d}+ 2^{3m-5d}+ 2^{3m-7d}-
2^{3m-8d})+2^n-2^{n-d}-2^{n-4d}-2^{n-6d}+2^d+1)}{2^d+1}}
$ \\[2mm]
\hline $\s{2^n}$&$\s{0}$&$\s{1}$
\\[2mm]
\hline
\end{tabular}
\end{center}
where $\mu=(-1)^{m/d}$.

\end{theo}
\begin{proof}
See \textbf{Appendix C.}
\end{proof}

In order to give the correlation distribution among the sequences in
$\calF$, we need an easy observation.

\begin{lemma}\label{q-2}
For any given $\ga\in \bF_q^*$, when $(\al,\be)$ runs through
$\bF_{q}\times \bF_q$, the distribution of $S(\al,\be,\ga)$ is the
same as $S(\al,\be,1)$.
\end{lemma}

As a consequence of Theorem \ref{value dis T}, Theorem \ref{value
dis S} and Lemma \ref{q-2}, we could give the correlation
distribution amidst the sequences in $\calF$.

\begin{theo}\label{cor dis}
Let $1\leq k\leq n-1$ and
$k\notin\left\{\frac{n}{6},\frac{n}{4},\frac{n}{2},\frac{3n}{4},\frac{5n}{6}\right\}$.
The collection $\calF$ defined in (\ref{def F}) is a family of
$2^{2n}$ ($n/d$ is odd) or $2^{2n}+2^n+1$ ($n/d$ is even) binary
sequences with period $q-1$. Its correlation distribution is given
as follows.

  (i). For the case $n/d$ is odd (that is, $d'=d$),
\begin{center}
\begin{tabular}{|c|c|}
\hline
values  & multiplicity \\[2mm]
\hline $\s{2^{(n+d)/2}-1}$  &
$\frac{(2^{n-d-1}+2^{(n-d-2)/2})(2^{n+2d}-2^{n}-2^{n-d}+2^{2d})(2^{3n}-2^{n+1})}{2^{2d}-1}$
\\[2mm]
\hline $
\s{-2^{(n+d)/2}-1}$&$\frac{(2^{n-d-1}-2^{(n-d-2)/2})(2^{n+2d}-2^{n}-2^{n-d}+2^{2d})(2^{3n}-2^{n+1})}{2^{2d}-1}$
\\[2mm]
\hline
$\s{2^{(n+3d)/2}-1}$&$\frac{(2^{n-3d-1}+2^{(n-3d-2)/2})(2^{n-d}-1)(2^{3n}-2^{n+1})}{2^{2d}-1}$
\\[2mm]
\hline
$\s{-2^{(n+3d)/2}-1}$&$\frac{(2^{n-3d-1}-2^{(n-3d-2)/2})(2^{n-d}-1)(2^{3n}-2^{n+1})}{2^{2d}-1}$
\\[2mm]
\hline
$\s{-1}$&$\scriptstyle{(2^{2n}-2^{2n-d}+2^{2n-4d}+2^n-2^{n-d}-2^{n-3d}+1)(2^{3n}-2^{n+1})+2^n}$
\\[2mm]
\hline $\s{2^n-1}$&$\s{2^{2n}+2^n}$
\\[2mm]
\hline
\end{tabular}
\end{center}

(ii).  For the case $n/d\equiv 0\pmod 4$.

\begin{center}
\begin{tabular}{|c|c|}
\hline
values & multiplicity \\[2mm]
\hline $ \scriptstyle{2^m-1}$&  $
\frac{2^{2n}\left(2^{n+6d}-2^{n+4d}-2^{n+d}+ 2^{m+5d}-
2^{m+4d}+2^{6d}\right)(2^{2n-1}+2^{3m-1}-2^n-2^m+1)}{(2^d+1)(2^{2d}-1)(2^{3d}+1)}
$
\\[2mm]
\hline $\scriptstyle{-2^m-1}$&  $
\frac{2^{2n}(2^{n-1}-2^{m-1})(2^n-2)\left(2^{n+6d}-2^{n+4d}-2^{n+d}+
2^{m+5d}- 2^{m+4d}+2^{6d}\right)}{(2^d+1)(2^{2d}-1)(2^{3d}+1)} $
\\[2mm]
\hline

 $ \scriptstyle{{2}^{m+d}-1}$&  $
\frac{2^{2n}(2^{n-2d-1}+2^{m-d-1})(2^n-2)(2^{n+3d}+2^{n+2d}-2^n-2^{n-d}-2^{n-2d}-
2^{m+3d}+ 2^{m}+2^{3d})}{(2^d+1)^2(2^{2d}-1)} $
\\[2mm]
\hline

 $\scriptstyle{-{2}^{m+d}-1}$&  $
\frac{2^{2n}(2^{n+3d}+2^{n+2d}-2^n-2^{n-d}-2^{n-2d}- 2^{m+3d}+
2^{m}+2^{3d})(2^{2n-2d-1}-2^{3m-d-1}-2^{n-2d}+2^{m-d}+1)}{(2^d+1)^2(2^{2d}-1)}
$
\\[2mm]
\hline

 $ \scriptstyle{{2}^{m+2d}-1}$&$
 \frac{
2^{2n}(2^{m-d}+1)(2^{m+d}+2^m-2^{m-2d}-
2^d)(2^{2n-4d-1}+2^{3m-2d-1}-2^{n-4d}-2^{m-2d}+1)}{(2^d+1)^3(2^{d}-1)}
$
\\[2mm]
\hline $\scriptstyle{-{2}^{m+2d}-1}$&$
 \frac{
2^{2n}(2^{m-d}+\mu)(2^{m+d}+2^m-2^{m-2d}-\mu
2^d)(2^{n-4d-1}-2^{m-2d-1})(2^{n}-2)}{(2^d+1)^3(2^{d}-1)} $
\\[2mm]
\hline
 $ \scriptstyle{{2}^{m+3d}-1}$&$
\frac{2^{2n}(2^{m-2d}-1)(2^{m-d}+1)(2^{n-6d-1}+2^{m-3d-1})(2^{n}-2)}{(2^d+1)(2^{2d}-1)(2^{3d}+1)}
$
\\[2mm]

\hline
 $\scriptstyle{-{2}^{m+3d}-1}$&$
\frac{2^{2n}(2^{m-2d}-1)(2^{m-d}+1)(2^{2n-6d-1}-2^{3m-3d-1}-2^{n-6d}+2^{m-3d}+1)}{(2^d+1)(2^{2d}-1)(2^{3d}+1)}
$
\\[2mm]
\hline $\scriptstyle{-1}$&
$\s{\frac{2^{2n}(2^n-2)\left(2^{2n}+2^{2n-9d}-
2^{3m}+
2^{3m-d}+ 2^{3m-3d}- 2^{3m-5d}- 2^{3m-7d}+
2^{3m-8d}+2^n-2^{n-d}-2^{n-4d}-2^{n-6d}+2^d+1\right)}{2^d+1}}
$ \\[2mm]
\hline $\scriptstyle{2^n-1}$&$\scriptstyle{2^{2n}}$
\\[2mm]
\hline
\end{tabular}
\end{center}

(iii). For the case $n/d\equiv 2\pmod 4$,
\begin{center}
\begin{tabular}{|c|c|}
\hline
values & multiplicity \\[2mm]
\hline $ \scriptstyle{2^m-1}$& $
\frac{2^{2n}(2^{n-1}+2^{m-1})(2^n-2)\left(2^{n+6d}-2^{n+4d}-2^{n+d}-
2^{m+5d}+ 2^{m+4d}+2^{6d}\right)}{(2^d+1)(2^{2d}-1)(2^{3d}+1)} $
\\[2mm]
\hline $\scriptstyle{- 2^m-1}$&  $
\frac{2^{2n}\left(2^{n+6d}-2^{n+4d}-2^{n+d}- 2^{m+5d}+
2^{m+4d}+2^{6d}\right)(2^{2n-1}-2^{3m-1}-2^n+2^m+1)}{(2^d+1)(2^{2d}-1)(2^{3d}+1)}
$
\\[2mm]
\hline

 $ \scriptstyle{{2}^{m+d}-1}$& $
\frac{2^{2n}(2^{n+3d}+2^{n+2d}-2^n-2^{n-d}-2^{n-2d}+ 2^{m+3d}-
2^{m}+2^{3d})(2^{2n-2d-1}+2^{3m-d-1}-2^{n-2d}-2^{m-d}+1)}{(2^d+1)^2(2^{2d}-1)}
$
\\[2mm]
\hline

 $\scriptstyle{-{2}^{m+d}-1}$&  $
\frac{2^{2n}(2^{n-2d-1}-2^{m-d-1})(2^n-2)(2^{n+3d}+2^{n+2d}-2^n-2^{n-d}-2^{n-2d}+
2^{m+3d}- 2^{m}+2^{3d})}{(2^d+1)^2(2^{2d}-1)} $
\\[2mm]
\hline

 $ \scriptstyle{{2}^{m+2d}-1}$&$
 \frac{
2^{2n}(2^{m-d}-1)(2^{m+d}+2^m-2^{m-2d}+
2^d)(2^{n-4d-1}+2^{m-2d-1})(2^{n}-2)}{(2^d+1)^3(2^{d}-1)} $
\\[2mm]
\hline $\scriptstyle{-{2}^{m+2d}-1}$&$
 \frac{
2^{2n}(2^{m-d}-1)(2^{m+d}+2^m-2^{m-2d}+
2^d)(2^{2n-4d-1}-2^{3m-2d-1}-2^{n-4d}+2^{m-2d}+1)}{(2^d+1)^3(2^{d}-1)}
$
\\[2mm]
\hline
 $ \scriptstyle{{2}^{m+3d}-1}$&$
\frac{2^{2n}(2^{m-2d}+1)(2^{m-d}-1)(2^{2n-6d-1}-2^{3m-3d-1}-2^{n-6d}+2^{m-3d}+1)}{(2^d+1)(2^{2d}-1)(2^{3d}+1)}
$
\\[2mm]

\hline
 $\scriptstyle{-{2}^{m+3d}-1}$&$
\frac{2^{2n}(2^{m-2d}+1)(2^{m-d}-1)(2^{n-6d-1}-2^{m-3d-1})(2^{n}-2)}{(2^d+1)(2^{2d}-1)(2^{3d}+1)}
$
\\[2mm]
\hline $\scriptstyle{-1}$&$
\s{\frac{2^{2n}(2^n-2)\left(2^{2n}+2^{2n-9d}+
2^{3m}-
2^{3m-d}- 2^{3m-3d}+ 2^{3m-5d}+ 2^{3m-7d}-
2^{3m-8d}+2^n-2^{n-d}-2^{n-4d}-2^{n-6d}+2^d+1\right)}{2^d+1}}
$ \\[2mm]
\hline $\scriptstyle{2^n-1}$&$\scriptstyle{2^{2n}}$
\\[2mm]
\hline
\end{tabular}
\end{center}

\end{theo}

\begin{proof}
see \textbf{Appendix C.}
\end{proof}
\begin{table}[ht!]\caption{\textbf{
Some Sequence Families with Low Correlation}}
\begin{center}
\begin{tabular}{|c|c|c|c|c|}
\hline
sequence  &$n$& period &family size &$C_{\mathrm{max}}$ \\[2mm]
\hline
Gold\cite{Gold}& odd &$2^n-1$& $2^n+1$ &$2^{(n+1)/2}+1$\\[1mm]
\hline
Rothaus\cite{Roth}&even&$2^n-1$& $2^{2n}+2^n+1$ &$2^{(n+3)/2}+1$\\[1mm]
\hline
Yu and Gong\cite{Yu Gon}, $1< \rho\leq \frac{n-1}{2}$&odd&$2^n-1$& $2^{n\rho}$ &$2^{(n+2\rho-1)/2}+1$\\[1mm]

\hline
Generalized Kasami(Large Set)\cite{Kasa3},\cite{Xia Zen}&even&$2^n-1$& $2^{3n/2}+1$ &$2^{(n+2)/2}$\\[1mm]

\hline
$\calF$ with $d=1$& odd &$2^n-1$& $2^{2n}+2^n+1$ &$2^{(n+3)/2}+1$\\[1mm]
\hline
\end{tabular}
\end{center}
\end{table}

\section{Appendix A}
{\it \textbf{Proof of Lemma \ref{rank}}}:

 For $Y=(y_1,\cdots,y_s)\in \bF_{q_0}^s$,
$y=y_1v_1+\cdots+y_sv_s\in \bF_q$, we know that
\begin{equation}\label{bil form1}
F_{\alpha,\beta}(X+Y)-F_{\alpha,\beta}(X)-F_{\alpha,\beta}(Y)=2XH_{\alpha,\beta}Y^T
\end{equation}
is equal to
\begin{equation}\label{bil form2}
f_{\alpha,\beta}(x+y)-f_{\alpha,\beta}(x)-f_{\alpha,\beta}(y)=\Tra_{d}^n\left(y^{2^{3k}}(\alpha^{2^{3k}}
x^{2^{6k}}+\be^{2^{3k}} x^{2^{4k}}+\beta^{2^{2k}} x^{2^{2k}}+\al
x)\right)
\end{equation}

 Let
\begin{equation}\label{def phi}
\phi_{\al,\be}(x)=\alpha^{2^{3k}} x^{2^{6k}}+\be^{2^{3k}}
x^{2^{4k}}+\beta^{2^{2k}} x^{2^{2k}}+\al x.
\end{equation}
 Therefore,
\[{\setlength\arraycolsep{2pt}
\begin{array}{lcl}
r_{\al,\be}=r& \Leftrightarrow&\text{the number of common solutions of}\;XH_{\alpha,\beta}Y^T=0\;\text{for all}\;Y\in \bF_{q_0}^s\;\text{is}\; q_0^{s-r}, \\[2mm]
& \Leftrightarrow&\text{the number of common solutions of}\;\Tra_{d}^n\left(y^{2^{3k}}\cdot\phi_{\al,\be}(x)\right)=0\;\text{for all}\;y\in \bF_q\;\text{is}\; q_0^{s-r}, \\[2mm]
&\Leftrightarrow&\phi_{\al,\be}(x)=0\;\text{has}\; q_0^{s-r}\;
\text{solutions in}\; \bF_q.
\end{array}
}
\]

For a fixed algebraic closure $\bF_{2^\infty}$ of $\bF_2$, since the
degree of $2^{2k}$-linearized polynomial $\phi_{\al,\be}(x)$ is
$2^{6k}$ and $\phi_{\al,\be}(x)=0$ has no multiple roots in
$\bF_{2^{\infty}}$ (this fact follows from
$\phi'_{\al,\be}(x)=\al\in \bF_q^*$),  then the zeroes of
$\phi_{\al,\be}(x)$ in $\bF_{2^\infty}$, say $V$, form an
$\bF_{2^{2k}}$-vector space of dimension 3. Note that
$\gcd(n,2k)=d'$. Then $V\cap \bF_{2^n}$ is a vector space on
$\bF_{2^{\gcd(n,2k)}}=\bF_{2^{d'}}$ of dimension less that or equal
to 3 since any elements in $\bF_{2^n}$ which are linear independent
over $\bF_{2^{d'}}$ are also linear independent over
$\bF_{2^{2k}}$(see \cite{Trac}, Lemma 4).

\begin{itemize}
  \item [(i).] For the case $n/d$ is odd, then $d'=d$ and $r_{\al,\be}\geq
  s-3$. Note that $r_{\al,\be}$ must be even. Hence the possible
  values of $r_{\al,\be}$ are $s-1$ and $s-3$.
  \item [(ii).] For the case $n/d$ is even, then $d'=2d$ and the possible
  values of $r_{\al,\be}$ are $s$, $s-2$, $s-4$ and $s-6$.
\end{itemize}

 $\square$

 {\it\textbf{Proof of Lemma \ref{moment}}}: (i). We observe that
\[ {\setlength\arraycolsep{2pt}
\begin{array}{ll}
&\sum\limits_{\al,\be\in
\bF_q}T(\al,\be)=\sum\limits_{\al,\be\in\bF_q}\sum\limits_{x\in
\bF_q}(-1)^{\Tra_1^n(\al x^{2^{3k}+1}+\be x^{2^k+1})}\\[3mm]
&\quad\quad=\sum\limits_{x\in\bF_q}\sum\limits_{\al\in
\bF_{q}}(-1)^{\Tra_1^n(\al x^{2^{3k}+1})}\sum\limits_{\be\in
\bF_q}(-1)^{\Tra_1^n(\be
x^{2^k+1})}=q\cdot\sum\limits_{\stackrel{\al\in
\bF_{q}}{x=0}}(-1)^{\Tra_1^n(\al x^{2^{3k}+1})}=2^{2n}.
\end{array}
}
\]
(ii). We can calculate
\[
{ \setlength\arraycolsep{2pt}
\begin{array}{lll}
\sum\limits_{\al,\be\in\bF_q}T(\al,\be)^2&=&\sum\limits_{x,y\in
\bF_q}\sum\limits_{\al\in
\bF_{q}}(-1)^{\Tra_1^n\left(\al\left(x^{2^{3k}+1}+y^{2^{3k}+1}\right)\right)}\sum\limits_{\be\in
\bF_q}(-1)^{\Tra_1^n\left(\be\left(x^{2^k+1}+y^{2^k+1}\right)\right)}\\[2mm]
&=&M_2\cdot 2^{2n}\end{array}}
\] where
$M_2$ is the number of solutions to the equation
\begin{eqnarray}\label{def 2nd}
\left\{
\begin{array}{ll}
 x^{2^{3k}+1}+y^{2^{3k}+1}=0&\\[2mm]
 x^{2^k+1}+y^{2^k+1}=0&
 \end{array}
 \right.
\end{eqnarray}

If $xy=0$ satisfying (\ref{def 2nd}), then $x=y=0$. Otherwise
$(x/y)^{2^{3k}+1}=(x/y)^{2^k+1}=1$ which yields that
$(x/y)^{2^{2k}-1}=1$. Denote by $x=ty$.  Since $\gcd(2k,n)=d'$, then
$t\in \bF_{2^{d'}}^*$.
\begin{itemize}
  \item If $d'=d$, then $t\in \bF_{2^d}^*$ and (\ref{def 2nd}) is
  equivalent to $x^2+y^2=0$, that is, $x=y$ which has $2^n-1$ solutions in $\bF_q^*$.  Therefore
  \[
  M_2=1+(2^n-1)=2^n.
  \]
  \item If $d'=2d$, then by (\ref{rel d d'}) we get (\ref{def 2nd}) is equivalent to
  $x^{2^d+1}+y^{2^d+1}=0$. Then we have $t^{2^d+1}=1$ which has
  $2^d+1$ solutions in $\bF_{2^{2d}}^*$. Therefore
  \[M_2=(2^d+1)(2^n-1)+1=2^{n+d}+2^n-2^d.\]
\end{itemize}

(iii). We have
\[\sum\limits_{\al,\be\in\bF_q}T(\al,\be)^3=M_3\cdot 2^{2n}\quad \text{where}\]
\begin{eqnarray}\label{num M}
&&M_3=\#\left\{(x,y,z)\in
\bF_q^3\left|x^{2^{3k}+1}+y^{2^{3k}+1}+z^{2^{3k}+1}=0,x^{2^k+1}+y^{2^k+1}+z^{2^k+1}=0\right.\right\}\nonumber\\[2mm]
&&\quad\;=M_2+T'\cdot(2^n-1)
\end{eqnarray}
and $T'$ is the number of $\bF_q$-solutions of
\begin{equation}\label{def T'}
\left\{
\begin{array}{ll}
x^{2^{3k}+1}+y^{2^{3k}+1}+1=0&\\[2mm]
x^{2^k+1}+y^{2^k+1}+1=0.&
\end{array}
\right.
\end{equation}
Canceling $y$ we have
$\left(x^{2^{3k}+1}+1\right)^{2^k+1}=\left(x^{2^{k}+1}+1\right)^{2^{3k}+1}$
which is equivalent to
\[(x^{2^{4k}}+x)(x^{2^k}+x^{2^{3k}})=0.\]
Therefore $x^{2^{4k}}=x$ or $x^{2^k}=x^{2^{3k}}$.
\begin{itemize}
  \item If $n/d\equiv 2\pmod 4$, then $x\in \bF_{2^{2d}}$ and symmetrically $y\in \bF_{2^{2d}}$. Hence
   (\ref{def T'}) is equivalent
  to $x^{2^d+1}+y^{2^d+1}+1=0$ which is the well-known Hermitian curve on $\bF_{2^{2d}}$.
  If follows that $T'=2^{3d}-2^d$.
  \item If $n/d\equiv 0\pmod 4$, then $x\in \bF_{2^{4d}}$ and hence $y\in \bF_{2^{4d}}$. In this case
  $\left(x^{2^k+1}+y^{2^k+1}+1\right)^{2^{3k}}=x^{2^{3k}+1}+y^{2^{3k}+1}+1$ and then
   (\ref{def T'}) is
  equivalent to $x^{2^d+1}+y^{2^d+1}+1=0$ which is a minimal curve on $\bF_{2^{4d}}$ with genus $2^{d-1}(2^d-1)$.
  Hence
  $T'=2^{4d}+1-2^d(2^d-1)2^{2d}-(2^d+1)=2^{3d}-2^d$ (Note that there are exactly $2^d+1$ rational infinite places on
  the curve $x^{2^d+1}+y^{2^d+1}+1=0$ defined on $\bF_q$).
\end{itemize}
Anyway,
$M_3=(2^{n+d}+2^n-2^d)+(2^n-1)(2^{3d}-2^d)=2^{n+3d}+2^n-2^{3d}$.

(iv). We can calculate
\[\sum\limits_{\al,\be,\ga\in\bF_q}S(\al,\be,\ga)^3=L_3\cdot 2^{3n}\quad \text{where}\;L_3\;\text{is the number of solutions to}\]
\begin{equation}\label{num L3}
x^{2^{3k}+1}+y^{2^{3k}+1}+z^{2^{3k}+1}=x^{2^k+1}+y^{2^k+1}+z^{2^k+1}=x+y+z=0
\end{equation}
\begin{itemize}
  \item If $xyz=0$, we may assume $x=0$ and then $y=z$ which gives $2^n$ solutions. So is $y=0$ or $z=0$. Note that $x=y=z=0$ has been counted 3 times. Hence
   there are exactly $3\cdot 2^n-2$ solutions to (\ref{num L3}) satisfying $xyz=0$.
  \item If $xyz=0$, then $L_3$ is equal to $2^n-1$ times the number of solutions to
  \begin{equation}\label{red s3}
   x^{2^{3k}+1}+y^{2^{3k}+1}+1=x^{2^k+1}+y^{2^k+1}+1=x+y+1=0
  \end{equation}
  with $xy\neq 0$. Then we have $x^{2^{3k}+1}+(x+1)^{2^{3k}}+1=x^{2^{k}+1}+(x+1)^{2^k}+1=0$ which is equivalent to $x^{2^{3k}}=x^{2^k}=x$.
  Hence $x\in \bF_{2^{3k}}\cap \bF_{2^{k}}\cap \bF_{2^{n}}^*=\bF_{2^{d}}^*$. Since $y\neq 0$, then $x\neq 1$. Therefore (\ref{red s3}) has
  $2^d-2$ solutions with $xy\neq 0$ and then $L_3=3\cdot 2^n-2+(2^d-2)(2^n-1)=2^{n+d}+2^n-2^d$.

\end{itemize}

 {\it\textbf{Proof of Lemma \ref{Artin}}}:
We get that
\[
\begin{array}{rcl}
qN&=&\sum\limits_{\om\in \bF_q}\sum\limits_{x,y\in
\bF_q}(-1)^{\Tra_1^n\left(\om\left(\al x^{2^{3k}+1}+\be x^{2^k+1}+y^{2^d}+y\right)\right)}\\[2mm]
&=&q^2+\sum\limits_{\om\in \bF_q^*}\sum\limits_{x\in
\bF_q}(-1)^{\Tra_1^n\left(\om\left(\al x^{2^{3k}+1}+\be
x^{2^k+1}\right)\right)} \sum\limits_{y\in
\bF_q}(-1)^{\Tra_1^n\left(y^{2^d}\left(\om^{2^d}+\om\right)\right)}\\[2mm]
&=&q^2+q\sum\limits_{\om\in \bF_{q_0}^*}\sum\limits_{x\in
\bF_q}(-1)^{\Tra_1^n\left(\om\left(\al x^{2^{3k}+1}+\be
x^{2^k+1}\right)\right)}\\[2mm]
&=&q^2+q\sum\limits_{\om\in \bF_{q_0}^*}T(\om\al,\om\be)
\end{array}
\]
where the 3-rd equality follows from that the inner sum is zero
unless $\om^{2^d}+\om=0$, i.e. $\om\in \bF_{q_0}$.

For any $\om\in \bF_{q_0}^*$,  choose $t\in \bF_{2^{2d}}$ such that
$t^{2^d+1}=\om$. Then $t^{2^{3k}+1}=t^{2^k+1}=\om$. As a consequence
$T({\om\al,\om\be})=T(\al,\be).$ Hence $N=2^n+(2^d-1)\cdot
T(\al,\be)$. $\square$

\section{Appendix B}

{\it\textbf{Proof of Theorem \ref{value dis T}}} (ii):

 In the case $n/d$ is even
(that is, $d'=2d$), we have $r_{\al,\be}=s,s-2,s-4$
or $s-6$ for $(\al,\be)\neq (0,0)$. According to Lemma \ref{qua} and
Lemma \ref{reduce num}, we get that if $r_{\al,\be}=s-i$, then
$T(\al,\be)=(-1)^{m/d+\frac{i}{2}}2^{\frac{n+id}{2}}$.

Combining Lemma \ref{rank} and Lemma \ref{moment} we have
\begin{equation}\label{par sum02}
 n_{0}+n_{2}+n_{4}+n_{6}=2^{2n}-1
 \end{equation}

\begin{equation}\label{par sum12}
 n_{0}-2^d\cdot n_{2}+2^{2d}\cdot n_{4}-2^{3d}\cdot n_6=(-1)^{m/d} 2^m(2^{n}-1)
 \end{equation}

 \begin{equation}\label{par sum22}
 n_{0}+2^{2d}\cdot n_{2}+2^{4d}\cdot n_{4}+2^{6d}\cdot n_{6}=2^n(2^d+1)(2^n-1).
 \end{equation}

  \begin{equation}\label{par sum32}
 n_{0}-2^{3d}\cdot n_{2}+2^{6d}\cdot n_{4}-2^{9d}\cdot n_{6}=(-1)^{m/d} 2^{m+3d}(2^{n}-1).
 \end{equation}
Solving the system of equations consisting of (\ref{par
sum02})--(\ref{par sum32}) yields the result.
$\square$

\section{Appendix C}

{\it\textbf{Proof of Theorem \ref{value dis S}}} : Define
\[\Xi=\left\{(\al,\be,\ga)\in \bF_q^3\left|S(\al,\be,\ga)=0 \right.\right\}\]
and $\xi=\big{|}\Xi\big{|}$.

 Recall $n_i,H_{\al,\be},,r_{\al,\be},A_{\ga}$ in Section 1 and
$N_{i,\veps},n_{i,\veps,}$ in the proof of Lemma \ref{rank}.
\begin{itemize}
  \item In the case $n/d$ is odd, from Lemma \ref{det gamma}, Lemma \ref{rank}(i) and Lemma \ref{moment}(iv) we have
\begin{equation}\label{n/d odd 1}
n_1+n_3=2^{2n}-1
\end{equation}

\begin{equation}\label{n/d odd 2}
n_1+2^{2d}\cdot n_3=2^{n-d}(2^d+1)(2^n-1).
\end{equation}

These two equations yield
\[n_1=\frac{(2^{n}-1)(2^{n+2d}-2^n-2^{n-d}+2^{2d})}{2^{2d}-1},\quad n_3=\frac{(2^{n-d}-1)(2^n-1)}{2^{2d}-1}.\]

From Lemma \ref{det gamma} we get that
\begin{equation}\label{value xi1}
{\setlength\arraycolsep{2pt}
\begin{array}{lll}
\xi&=2^{n}-1+(2^{n}-2^{n-d})n_{1}+(2^{n}-2^{n-3d})n_{3}&\\[2mm]
&=(2^n-1)(2^{2n}-2^{2n-d}+2^{2n-4d}+2^n-2^{n-d}-2^{n-3d}+1)
\end{array}
}
\end{equation}
  \item In the case $n/d$ is even, from
Lemma \ref{det gamma}, Lemma \ref{rank}(ii), Lemma \ref{reduce num} and Lemma \ref{value dis T}(ii) we get that
\begin{equation}\label{value xi2}
{\setlength\arraycolsep{2pt}
\begin{array}{lll}
\xi&=2^{n}-1+(2^{n}-2^{n-2d})n_{2,1}+(2^{n}-2^{n-4d})n_{4,-1}+(2^n-2^{n-6d})n_{6,1}&\\[2mm]
&\begin{array}{ll} =(2^n-1)\left[1+(2^{2n}+2^{2n-9d}-\veps
2^{3m}+\veps
2^{3m-d}+\veps 2^{3m-3d}\right.&\\[1mm]
\quad\left.-\veps 2^{3m-5d}- \veps 2^{3m-7d}+\veps
2^{3m-8d}+2^n-2^{n-d}-2^{n-4d}-2^{n-6d})/(2^d+1)\right]&
\end{array}
\end{array}
}
\end{equation}
\end{itemize}
By Lemma \ref{det gamma} we get the value distribution of $S(\al,\be,\ga)$ and then the weight distribution of $\cC_2$.
 $\square$

{\it\textbf{Proof of Theorem \ref{cor dis}}} : For any possible
value $\ka$ and $1\leq i, j\leq 2$, define $M_{\ka}(\calF_i,
\calF_j)$ to be the frequency of $\ka$ in correlation values
 between any two sequences in $\calF_i$ and $\calF_j$ by any shift, respectively.  Then the correlation distribution of sequences in
 $\calF$
could be obtained if we can calculate all of the $M_{\ka}(\calF_i,
\calF_j)$ . We will deal with it case by case.
\begin{itemize}
  \item The correlation function between
$a_{\al_1,\be_1}$ and $a_{\al_2,\be_2}$ by a shift $\tau$ ($0\leq
\tau\leq q-2$) is
\[
\begin{array}{ll}
&C_{(\al_1,\be_1),(\al_2,\be_2)}(\tau)=\sum\limits_{\lambda=0}^{q-2}(-1)^{a_{\al_1,\be_1}({\lambda})-
a_{\al_2,\be_2}({\lambda+\tau})}\\[2mm]
&\qquad =\sum\limits_{\lambda=0}^{q-2}(-1)^{\Tra_1^n(\al_1
  \pi^{\lambda(2^{3k}+1)}+\be_1
  \pi^{\lambda(2^k+1)}+\pi^{\lambda})-\Tra_1^n(\al_2 \pi^{(\lambda+\tau)(2^{3k}+1)}+\be
  \pi^{(\lambda+\tau)(2^k+1)}+\pi^{\lambda+\tau})}\\[2mm]
  &\qquad = S(\al',\be',\ga')-1
  \end{array}
\]
 where
 \begin{equation}\label{coe cor}
 \al'=\al_1-\al_2 \pi^{\tau(2^{3k}+1)},\quad
 \be'=\be_1-\be_2\pi^{\tau(2^k+1)},\quad \ga'=1-\pi^{\tau}.
 \end{equation}

 Fix $(\al_2,\be_2)\in
\bF_q\times \bF_q$, when $(\al_1,\be_1)$ runs through $\bF_q\times
\bF_q$ and $\tau$ takes values from $0$ to $q-2$, $(\al',\be',\ga')$
runs through $\bF_q\times
\bF_q\times\left\{\bF_{q}\big{\backslash}\{1\}\right\}$ exactly one
time.

For any possible value $\kappa$ of $S(\al,\be,\ga)$, define

\begin{equation}\label{def sk}
s_{\kappa}=\#\left\{(\al,\be,\ga)\in \bF_q\times \bF_q\times
\bF_q\,\displaystyle{|}\,S(\al,\be,\ga)=\kappa+1\right\}
\end{equation}

\begin{equation}\label{def sk1}s^1_{\kappa}=\#\left\{(\al,\be)\in \bF_q\times
\bF_q\,\big{|}\,S(\al,\be,1)=\kappa+1\right\}
\end{equation}
 and
\begin{equation}\label{def tk}t_{\kappa}=\#\left\{(\al,\be)\in \bF_q\times
\bF_q\,\displaystyle{|}\,T(\al,\be)=\kappa+1\right\}.
\end{equation}

By Lemma \ref{q-2} we have
\begin{equation}\label{rel sk1 sk tk}
s_{\kappa}^1=\frac{1}{2^n-1}\times (s_{\kappa}-t_{\kappa}).
\end{equation}

Hence we get
\[M_{\kappa}(\calF_1,\calF_1)=2^{2n}\cdot \left(s_{\kappa}-s_{\kappa}^1\right)=2^{2n}\cdot\left(\frac{2^n-2}{2^n-1}\cdot s_{\kappa}+\frac{1}{2^n-1}\cdot t_{\kappa}\right).\]

  \item For the case $n/d$ is odd.  The cross correlation function between
$a_{\al_1,\be_1}$ and $a_{\al_2}$ by a shift $\tau$ ($0\leq \tau\leq
q-2$) is
\[
\begin{array}{ll}
&C_{(\al_1,\be_1),\al_2}(\tau)=\sum\limits_{\lambda=0}^{q-2}(-1)^{a_{\al_1,\be_1}({\lambda})-
a_{\al_2}({\lambda+\tau})}\\[2mm]
&\qquad =\sum\limits_{\lambda=0}^{q-2}(-1)^{\Tra_1^n(\al_1
\pi^{\lambda(2^{3k}+1)}+\be_1
  \pi^{\lambda(2^k+1)}+\pi^{\lambda})-\Tra_1^n(\al_2\pi^{(\lambda+\tau)(2^{3k}+1)}+
  \pi^{(\lambda+\tau)(2^k+1)})}\\[2mm]
  &\qquad = S(\al',\be',1)-1
  \end{array}
\]
 where
 $
 \al'=\al_1- \al_2\pi^{\tau(2^{3k}+1)},
 \be'=\be_1-\pi^{\tau(2^k+1)}.
$

 Fix $0\leq \tau\leq q-2$ and $\al_2\in \bF_{q}$, when $(\al_1,\be_1)$ runs through
$\bF_q\times \bF_q$, $(\al',\be')$ runs through $\bF_q\times \bF_q$
exactly one time.

  The cross correlation function between
$a_{\al_1,\be_1}$ and $a$ by a shift $\tau$ ($0\leq \tau\leq q-2$)
is
\[
\begin{array}{ll}
&C_{(\al_1,\be_1)}(\tau)=\sum\limits_{\lambda=0}^{q-2}(-1)^{a_{\al_1,\be_1}({\lambda})-
a({\lambda+\tau})}\\[2mm]
&\qquad =\sum\limits_{\lambda=0}^{q-2}(-1)^{\Tra_1^n(\al_1
\pi^{\lambda(2^{3k}+1)}+\be_1
  \pi^{\lambda(2^k+1)}+\pi^{\lambda})-\Tra_1^n(\pi^{(\lambda+\tau)(2^{3k}+1)})}\\[2mm]
  &\qquad = S(\al',\be_1,1)-1
  \end{array}
\]
 where
 $
 \al'=\al_1- \pi^{\tau(2^{3k}+1)}.
$

 Fix $0\leq \tau\leq q-2$, when $(\al_1,\be_1)$ runs through
$\bF_q\times \bF_q$, $(\al',\be_1)$ runs through $\bF_q\times \bF_q$
exactly one time.

In total, by (\ref{def sk}) and (\ref{def sk1}) we get
\[M_{\kappa}(\calF_1,\calF_2)=M_{\kappa}(\calF_2,\calF_1)=2^n(2^n-1)\cdot s_{\kappa}^1+(2^n-1)\cdot s_{\kappa}^1=(2^n+1)(s_{\ka}-t_{\ka}).\]

  \item For the case $n/d$ is odd.  The cross correlation function between
$a_{\al_1}$ and $a_{\al_2}$ by a shift $\tau$ ($0\leq \tau\leq q-2$)
is
\[
\begin{array}{ll}
&C_{\al_1,\al_2}(\tau)=\sum\limits_{\lambda=0}^{q-2}(-1)^{a_{\al_1}({\lambda})-
a_{\al_2}({\lambda+\tau})}\\[2mm]
&\qquad
=\sum\limits_{\lambda=0}^{q-2}(-1)^{\Tra_1^n(\al_1\pi^{\lambda(2^{3k}+1)}+
  \pi^{\lambda(2^k+1)})-\Tra_1^n(\al_2\pi^{(\lambda+\tau)(2^{3k}+1)}+
  \pi^{(\lambda+\tau)(2^k+1)})}\\[2mm]
  &\qquad = T(\al',\be')-1
  \end{array}
\]
 where
 $
 \al'=\al_1- \al_2\pi^{\tau(2^{3k}+1)},
 \be'=1-\pi^{\tau(2^k+1)}.
$

When $(\al_1,\al_2)$ runs through $\bF_{q}\times \bF_q$ and $\tau$
takes value from $0$ to $q-2$, $(\al',\be')$ runs through
$\bF_q\times \bF_q\large{\backslash}\{1\}$ exactly $ q$ times.

The cross correlation function between $a_{\al_1}$ and $a$ by a
shift $\tau$ ($0\leq \tau\leq q-2$) is
\[
\begin{array}{ll}
&C_{\al_1}(\tau)=\sum\limits_{\lambda=0}^{q-2}(-1)^{a_{\al_1}({\lambda})-
a({\lambda+\tau})}\\[2mm]
&\qquad
=\sum\limits_{\lambda=0}^{q-2}(-1)^{\Tra_1^n(\al_1\pi^{\lambda(2^{3k}+1)}+
  \pi^{\lambda(2^k+1)})-\Tra_1^n(\pi^{(\lambda+\tau)(2^{3k}+1)})}\\[2mm]
  &\qquad = T(\al',1)-1
  \end{array}
\]
 where
 $
 \al'=\al_1- \pi^{\tau(2^{3k}+1)}.
$ For fixed $\tau$, $0\leq \tau\leq q-2$, when $\al_1$ runs through
$\bF_{q}$, then so is $\al'$.

The auto-correlation function of $a$ by a shift $\tau$ ($0\leq
\tau\leq q-2$) is
\[
\begin{array}{ll}
&C(\tau)=\sum\limits_{\lambda=0}^{q-2}(-1)^{a({\lambda})-
a({\lambda+\tau})}\\[2mm]
&\qquad =\sum\limits_{\lambda=0}^{q-2}(-1)^{\Tra_1^n(\pi^{\lambda(2^{3k}+1)})-\Tra_1^n(\pi^{(\lambda+\tau)(2^{3k}+1)})}\\[2mm]
  &\qquad = T(1- \pi^{\tau(2^{3k}+1)},0)-1.
  \end{array}
\]
When $\tau$ takes values from $0$ to $q-2$, then $1-
\pi^{\tau(2^{3k}+1)}$ runs through $\bF_q\backslash\{1\}$.

 Define
\begin{equation}\label{def tk0}
\displaystyle{t_{\ka}^0=\#\left\{\be\in
\bF_q\large{|}\;T(\al,0)=\ka+1\right\}}
\end{equation}
and
\begin{equation}\label{def tk1}
\displaystyle{t_{\ka}^1=\#\left\{\be\in
\bF_q\large{|}\;T(\al,1)=\ka+1\right\}}.
\end{equation}
Then a routine calculation shows that

\[
t_{\ka}^0=\left\{\begin{array}{cl}1, & \ka=2^n-1\\[2mm]
2^n-1,&\ka=-1.\end{array}\right.
\]
By Lemma \ref{red s3} we have
\[t_{\ka}^1=\frac{1}{2^n-1}(t_{\ka}-t_{\ka}^0).\]
By Inclusion-Exclusion principle, Lemma \ref{q-2},  (\ref{def tk})
and (\ref{def tk0}) we get
\[
M_{\kappa}(\calF_2,\calF_2)=\left(\frac{2^n-2}{2^n-1}(t_{\ka}-t_{\ka}^0)+t_{\ka}^0\right)\cdot
2^n.
\]

In total, sum up all the $M_{\ka}(\calF_i,\calF_i)$ for $1\leq
i,j\leq 2$ and the result follows from Theorem \ref{value dis T} and
Theorem \ref{value dis S}.
\end{itemize} $\square$

\section{Conclusion and Further Study}

\quad In this paper we have studied the exponential sums
$T(\al,\be)$ and $S(\al,\be,\ga)$ with $\al,\be,\ga\in \bF_{2^n}$.
After giving the value distribution of $T(\al,\be)$ and $S(\al,\be,\ga)$, we determine the weight
distributions of the cyclic codes $\cC_1$ and $\cC_2$ and the
correlation distribution among a family of sequences.

In particular, for the case $f(x)=x^{2^{lk}+1}+x^{2^k+1}$ with $l\geq 4$, we could get the possible values of
$\sum\limits_{x\in \bF_q}(-1)^{\Tr_1^n(f(x))}$ and $\sum\limits_{x\in \bF_q}(-1)^{\Tr_1^n(f(x)+\ga x)}$.
But the first three moment identities developed in Lemma \ref{moment}
is not enough to determine the value distribution. However, we could get the possible weights of the corresponding
cyclic codes. New machinery and technique should be invented to attack this problem.

In \cite{Joh Hel}, the cross correlation distribution between two binary m-sequences $\{s_t\}$ and $\{s_{dt}\}$ with $d=\frac{2^{2k}+1}{2^k+1}$ for
$m$ odd and $k=1$ has been determined. It seems that our techniques may
be useful for the case $d=\frac{2^{lk}+1}{2^k+1}$ for a quite general $l$ and $k$.

It turns out that the evaluation of $T(\al,\be)$ and $S(\al,\be,\ga)$ are very similar for binary and odd characteristic
when $n/d$ is even. But the result is quite different when $n/d$ is odd.
\section{Acknowledgements}
\quad The authors will thank the anonymous referees for their
helpful comments.


\begin{thebibliography}{800}

\bibitem[1]{Can Cha}A.~Canteaut, P.~Charpin, and H.~Dobbertin, ``Binary m-sequences
with three-valued cross-correlation: a proof of Welch¡¯s
conjecture," \emph{IEEE Trans. Inform. Theory,} vol. 46, no. 1, pp.
4--8, Jan. 2000.

\bibitem[2]{Cus}T.~Cusick and H.~Dobbertin, ``Some new three-valued
cross-correlation functions for binary m-sequences," \emph{IEEE
Trans. Inform. Theory,} vol. 42, no. 4, pp. 1238--1240, April 1996.


\bibitem[3]{Coul1}R.S.~Coulter, ``ON the evaluation of
a class of Weil sums in characteristic 2," \emph{New Zeal. Jour.
Math.,} vol. 28, pp. 171--184, 1999.





\bibitem[4]{Din Hel1} C.~Ding, T.~Helleseth, and K.Y.~Lam, ``Several classes of
binary sequences with three-level autocorrelation," \emph{IEEE
Trans. Inf. Theory,} vol. 45, no. 7, pp. 2606--2612, Nov. 1999.

\bibitem[5]{Din Hel2} C.~Ding, T.~Helleseth, and H. Martinsen, ``New families of
binary sequences with optimal three-valued autocorrelation,"
\emph{IEEE Trans. Inf. Theory,} vol. 47, no. 1, pp. 428--433, Jan.
2001.

\bibitem[6]{Dob Fel}H.~Dobbertin, P.~Felke, T.~Helleseth, and P.~Rosendahl, ``Niho
type cross-correlation functions via Dickson polynomials and
Kloosterman sums," \emph{IEEE Trans. Inform. Theory,} vol. 52, no.
2, pp. 613--627, Feb. 2006.


\bibitem[7]{Gold}R.~Gold, ``Maximal recursive sequences with 3-valued recursive cross-correlation functions,"
\emph{IEEE Trans. Inf. Theory,} vol. 14, no. 1, pp. 154--156, Jan.
1968.

\bibitem[8]{Hell1}T.~Helleseth, ``Some results about the cross-correlation function between two maximal linear
sequences, " \emph{Discrete Math.,} vol. 16, no. 3, pp. 209--232, 1976.

\bibitem[9]{Hell2}T.~Helleseth, ``Pairs of m-sequences with a six-valued crosscorrelation," \emph{Mathematical Mathematical
Properties of Sequences and Other Combinatorial Structures,}  J.-S. No, H.-Y. Song,
T. Helleseth, and P. V. Kumar, Eds.  pp. 1--6, Boston:  Kluwer Academic Publishers, 2003.

\bibitem[10]{Hel Kho}T.~Helleseth, A.~Kholosha, and G.J.~Ness, ``Characterization of
m-Sequences of lengths $2^{2k}-1$ and $2^k-1$ with three-valued
cross correlation," \emph{IEEE Trans. Inform. Theory,} vol. 53, no.
6, pp. 2236--2245, June 2007.

\bibitem[11]{Hel Kum}T.~Helleseth and P.V.~Kumar, ``Sequences with low
correlation," in \emph{Handbook of Coding Theory}, V. S. Pless and
W. C. Huffman, Eds.  Amsterdam, The Netherlands: North-Holland,
1998.

\bibitem[12]{Hol Xia}H.~Hollmann and Q.~Xiang, ``A proof of Welch and Niho conjectures on
cross-correlation of binary m-sequences," \emph{Finite Fields and
Their Appli.,} vol. 7, no. 2, pp. 253--286, 2001.


\bibitem[13]{Joh Hel}A. ~Johansen and T. ~Helleseth, ``A family of m-sequences with five-valued cross correlations",
\emph{IEEE Trans. Inf. Theory,} vol. 55, no. 2, pp. 880--887, Feb.
2009.

\bibitem[14] {Luo Lin Xin}J.~Luo, S.~Ling, C.~Xing, ``Cyclic codes and sequences from a class of Dembowski-Ostrom functions,"
preprint.


\bibitem[15] {Luo Tan Wan}J.~Luo, Y.~Tang, H.~Wang ``Exponential sums, cyclic codes and sequences: the binary Kasami case,"
preprint.

\bibitem[16]{Kasa1}T.~Kasami, ``Weight distribution of Bose-Chaudhuri-Hocquenghem codes,"  in Combinatorial Mathematics and Its Applications.
 R. C. Bose and T. A. Dowling, Eds. Chapel Hill, NC: Univerisity of
 North Carolina Press, 1969, pp. 335--357.


\bibitem[17]{Kasa2}T.~Kasami, ``Weight distribution formula for some class of cyclic
codes," Coordinated Sci. Lab., Univ. Illinois, Urabana-Champaign,
Tech. Rep. R-285(AD 635274), 1966.

\bibitem[18]{Kasa3}T.~Kasami, ``Weight enumerators for several classes of subcodes of the
2nd order Reed-Muller codes," \emph{Inf. Contr.,} vol. 18, pp. 369--394, 1971.

\bibitem[19]{Lid Nie}R.~Lidl, and H.~Niederreiter, Finite
Fields, Addison-Wesley,  Encyclopedia of Mathematics and its
Applications, vol. 20, 1983.

\bibitem[20]{Nes Hel}G. J. ~Ness and T. ~Helleseth, ``A new family of four-valued cross correlation between
m-sequences of different lengths," \emph{IEEE Trans. Inf. Theory,}
vol. 53, no. 11, pp. 4308--4313, Nov. 2007.


\bibitem[21]{Nes Hel2}G.J.~Ness and T.~Helleseth, ``Cross correlation of m-sequences of different lengths," \emph{IEEE Trans. Inf.
Theory,} vol. 52, no. 4,  pp. 1637--1648, Apri. 2006.

\bibitem[22]{Niho} Y.~Niho, ``Multivalued cross-correlation functions
between two maximal linear recursive sequences," Ph.D dissertation,
Univ. South.Calif., Los Angles, 1972.

\bibitem[23]{Part} K.G.~Paterson ``Applications of exponential sums
in communications theory," Extended Enterprise Laboratory, HP
Laboratories Bristol, HPL-1999-101, 13th Sept., 1999.



\bibitem[24]{Rose} P.~Rosendahl,  ``Niho type cross-correlation functions and related equations,"  Ph.D dissertation,
Depart. Math., Univ. Turku, Finland, 2004.

\bibitem[25]{Roth}O.S.~Rothaus, ``Modified Gold codes," \emph{IEEE Trans. Inf. Theory,} vol. 39,
no. 2, pp. 654--656, Mar. 1993.

\bibitem[26]{Sim Omu}M.K.~Simon, J.~Omura, R.~Scholtz, and K.~Levitt, ``Spread Spectrum
Communications," Rockville, MD: Computer Science, 1985, vol.I¨CIII.


\bibitem[27]{Trac} H.M.~Trachtenberg,  ``On the cross-correlation function of
maximal linear sequences,"  Ph.D dissertation, Univ. South. Calif.
Los Angles, 1970.


\bibitem[28]{Vand2}M.~Van Der Vlugt, ``Surfaces and the weight distribution of a family of codes," \emph{IEEE Trans. Inf. Theory,}
vol. 43, no. 4, pp. 1354--1360, Apri. 1997.

\bibitem[29]{Wolf}J.~Wolfmann, ``Weight distribution of some binary
primitive cylcic codes," \emph{IEEE Trans. Inf. Theory,} vol. 40,
no. 6, pp. 2068--2071, Jun. 2004.

\bibitem[30]{Xia Zen} Y.~Xia, X.~Zeng, and L.~Hu, ``The large set of $p$-ary Kasami sequences", preprint.


\bibitem[31]{Yu Gon}N.Y.~Yu and G.~Gong, ``A new binary sequence family with low
correlation and large size,", \emph{IEEE Trans. Inf. Theory,} vol.
52, no. 4, pp. 1624--1636, April 2006.

\bibitem[32]{Yu Gon2}N.Y.~Yu and G.~Gong,  ``New binary sequences with optimal autocorrelation magnitude,"
\emph{IEEE Trans. Inf. Theory,} vol. 54, no. 10, pp. 4771--4779,
Oct. 2008.


\bibitem[33]{Zen Hu Jia Yue Cao} X.~Zeng, L.~Hu, W.~Jiang, Q.~Yue and X.~Cao, ``Weight distribution of a $p$-ary cyclic code,"
see http://arxiv.org/abs/0901.2391, preprint.


\end{thebibliography}
\end{document}